\documentclass[reqno]{amsart}
\usepackage{amsaddr}
\usepackage{amsfonts, amsmath, amssymb, amsthm}

\usepackage{mathptmx}

\usepackage[top=3cm, bottom=3cm, left=3cm, right=3cm]{geometry}

\usepackage{multirow,multicol,diagbox}
\usepackage{algorithm}
\usepackage[table]{xcolor}

\usepackage[square,sort,comma,numbers]{natbib}

\newtheorem{theorem}{Theorem}
\newtheorem{cor}[theorem]{Corollary}
\newtheorem{lemma}[theorem]{Lemma}

\newtheorem{defn}[theorem]{Definition}

\theoremstyle{definition}

\theoremstyle{remark}

\usepackage{enumitem}

\usepackage[T1]{fontenc}
\DeclareMathAlphabet{\mathpzc}{OT1}{pzc}{m}{it}
\usepackage{yfonts}
\usepackage[sans]{dsfont}
\usepackage{txfonts}
\usepackage[mathscr]{euscript}
\usepackage{bbm}

\newcommand{\grp}[1]{\mathpzc{#1}}	
\newcommand{\set}[1]{\mathpzc{#1}}	

\newcommand{\cpx}{\mathbb C}			

\newcommand{\ff}[1]{{\mathbb F}_{#1}}		
\newcommand{\ffs}[1]{{\mathbb F}_{#1}^\star}	
\newcommand{\ffx}[1]{\ff{#1}[x]}		

\DeclareMathOperator{\image}{Im}
\newcommand{\im}[1]{\image(#1)}
\newcommand{\ims}[1]{\image(#1)\setminus \{0\}}

\DeclareMathOperator{\res}{Res}			
\DeclareMathOperator{\pres}{P-Res}		

\DeclareMathOperator{\Tr}{Tr}

\newcommand{\Trsub}[2]{\Tr_{#1/#2}}

\date{}

\begin{document}

\title{Permutation resemblance}
\author{Li-An Chen and Robert S. Coulter}
\address{Department of Mathematical Sciences\\
University of Delaware\\
Newark DE 19716, USA}

\begin{abstract}
Motivated by the problem of constructing bijective maps with low
differential uniformity, we introduce the notion of {\em permutation
resemblance} of a function, which looks to measure the distance a given map is
from being a permutation.
We prove several results concerning permutation resemblance and show how it
can be used to produce low differentially uniform bijections.
We also study the permutation resemblance of planar functions, which over
fields of odd characteristic are known not to be bijections and to have the
optimal differential uniformity.
\end{abstract}

\maketitle

\section{Introduction}

Throughout this paper $\grp{G}$ denotes a finite group of order $q$, written additively but
not necessarily abelian, and $\grp{G}^{\star}=\grp{G}\setminus\{0\}$. The finite field of order $q$ is denoted by
$\ff{q}$ and $\ffs{q} = \ff{q}\setminus\{0\}$. For a finite set $\set{S}$,
$\# \set{S}$ denotes the cardinality of $\set{S}$. 

Let $f:\grp{G}\to \grp{G}$.
The set of distinct images of $f$ is denoted by
$\im{f}:=\{f(x)\,:\, x\in\grp{G}\}$, and we write
$V(f):=\# \im{f}$.
If $S\subseteq\grp{G}$, then the set of images of $S$ under $f$ is denoted by
$f(S):=\{f(x)\,:\, x\in S\}$.
For $b\in \grp{G}$, the set of preimages of $b$ under $f$ is denoted by
$f^{-1}(b):=\{x\in \grp{G}\,:\, f(x)=b\}$, and we define the
\textit{uniformity of $f$} by \[u(f):=\max_{b\in \grp{G}} \#f^{-1}(b).\]

We call $f$ a \textit{permutation} over $\grp{G}$ if
$\im{f}=\grp{G}$ and we use $\Omega_{\grp{G}}$ to denote the set of all
permutation over $\grp{G}$.
As is well known, the set $\Omega_{\grp{G}}$ is a group under composition of functions.
Over a finite field $\ff{q}$, since every function can be uniquely represented
by a polynomial in $\ffx{q}$ of degree at most $q-1$, any polynomial that
represents a permutation over $\ff{q}$ under evaluation is called a
\textit{permutation polynomial (PP)} and we will often use this term instead
when working with $\ff{q}$. 

In this paper, we are interested in measuring how far a non-permutation
function $f$ over $\grp{G}$ is from being a permutation.
A standard measure is to consider $V(f)$, with the higher the value, the
closer the function is to being a permutation.
Here, we introduce a new notion for measuring this ``distance'', called
permutation resemblance.
Firstly, let $f$ and $h$ be two functions defined on $\grp{G}$.
The \textit{resemblance of $f$ to $h$} is defined by
\begin{equation*}
\res(f,h)=V(f-h). 
\end{equation*}
Clearly, $\res(f,f)=\# \{0\}=1$ for any function $f$.
Also, $\res(f,h)=\res(h,f)$, and $\res(f,h+c)=\res(f,h)$ for any constant
$c\in \grp{G}$. 
We may now introduce the central idea of this article.
\begin{defn} \label{presdefn}
For $f:\grp{G}\to \grp{G}$, we define the \textit{permutation resemblance of $f$} by
\begin{equation*}
\pres(f)=\min\{\res(f,h)\,:\, h\in\Omega_{\grp{G}}\}, 
\end{equation*}
or equivalently, by writing $f-h=-g$, 
\begin{equation*}
\pres(f)=\min\{V(g)\,:\, g+f\in\Omega_{\grp{G}}\}. 
\end{equation*}
\end{defn}
Intuitively, permutation resemblance measures the smallest number of different
shifts required to alter a function so that it becomes a permutation.
Put another way, one can view $\pres(f)$ as the minimum value of $V(g)$,
where $g$ runs through all functions over $\grp{G}$ for which
$g+f\in \Omega_{\grp{G}}$.
The higher the $\pres(f)$, the further $f$ is away from being a
permutation.
In particular, $f$ is a permutation if and only if $\pres(f)=\res(f,f)=1$.
On the other hand, $f$ is a constant if and only if $\pres(f)=\# \grp{G}$.
It is important to note that this measure is distinct from the common measure
already mentioned above. While there is a clear relation (even in the
definition), permutation resemblance seems to provide a more nuanced measure than
$V(f)$.
This can be seen, for instance, in our main result which in part shows how the
upper bound for permutation resemblance is inverse proportional to $V(f)$.
We note that, in general, for a given $f$, there is more than one $g$ for which
$g+f$ is a permutation and $V(g)=\pres(f)$.
We are aware of two papers -- \cite{carlet07} and \cite{panario11} -- where
similar notions were introduced. At the end of the paper we will discuss the
inequivalence of permutation resemblance to these earlier notions.

Non-trivial lower and upper bounds can be established for $\pres$, and these
bounds are non-trivial as soon as we move away from the two extremal
(and uninteresting) cases of constant functions and permutations.
In Section \ref{sc:lb_ub_pres} we establish such bounds.
We also prove that the two bounds are the same if and only if $f$ is a permutation, or its preimage distribution has a certain format. When $u(f)\neq q-V(f)+1$, the lower bound is still tight, and we give two classes of functions where the lower bound of Theorem \ref{th:lb_ub_pres} is met.
In Section \ref{sc:pres_inv}, we consider some widely used equivalence relations
on functions and how $\pres$ behaves under these relations.
In Section \ref{sc:pres_planar}, we consider the relation between
permutation resemblance and differential uniformity, which measures a function's
resistance to differential cryptanalysis when used as an S-box.
The problem of constructing permutation polynomials with low DU
has driven much research in information security in recent years
and our main motivation for introducing permutation resemblance is to provide
a new tool for approaching this problem.
The main reasons for why we believe $\pres$ is an important concept are
illustrated in this section. In particular, we point to
Theorem \ref{th:diffuni_bound}, which shows how one can change a
non-permutation into a
permutation while bounding how much differential uniformity increases.
The final section of the paper connects the upper bound in Theorem \ref{th:lb_ub_pres} with the values
\begin{equation}\label{eq:def_ns}
N_s(f)=\#\{(x_1,\dots,x_s)\in \grp{G}^{s}\,:\, f(x_i)=f(x_j),\,\forall 1\le i,j\le s; x_i= x_j\leftrightarrow i=j\}, 
\end{equation}
with $f:\grp{G}\to\grp{G}$ and $2\le s\le u(f)$.
There we show that $\#(\grp{G}\setminus\im{f})=q-V(f)$ can be written as a sum of $N_s(f)$. These numbers are known to be closely related to $V(f)$. In \cite{coulter14c}, Coulter and Senger gave a lower bound of $V(f)$ in terms of $N_s(f)$ for any given $2\le s\le u(f)$. Moreover, the case when $s=2$ is particularly useful for bounding $V(f)$ and other relavant quantities. One early example was given by Carlitz \cite{carlitz55} in 1955, where he showed a lower bound of $V(f)$ on average by writing $N_2(f)+q$ as a character sum. For more recent applications of $N_2(f)$ and $N_2(f)+q$, see \cite{coulter11,carlet21,kolsch22} for example.
Finally, we close the paper by showing that $\pres$ is not equivalent to other notions defined in \cite{carlet07} and \cite{panario11} that concern the bijectivity of the difference operators of functions over finite groups. 

\section{Establishing Lower and Upper Bounds for $\pres$}\label{sc:lb_ub_pres}

Our first result on $\pres$ establishes non-trivial lower and upper bounds
on $\pres$. 
\begin{theorem}\label{th:lb_ub_pres}
If $f:\grp{G}\to\grp{G}$, then 
\begin{equation}\label{eq:lb_ub_pres}
u(f)\le \pres(f)\le q-V(f)+1. 
\end{equation}
\end{theorem}
The lower bound is based on the following observation. 
\begin{lemma}\label{lm:g_1to1_on_preimg}
Let $f:\grp{G}\to\grp{G}$. If $g:\grp{G}\to\grp{G}$ satisfies that $g+f$ is a permutation, then $g$ is injective on each preimage set of $f$, i.e., for all $b\in\im{f}$, if $x,y\in f^{-1}(b)$ and $x\neq y$, then $g(x)\neq g(y)$. 
\end{lemma}
\begin{proof}
Let $f,g:\grp{G}\to\grp{G}$ such that $g+f$ is a permutation. If for some $b\in \im{f}$, $x,y\in f^{-1}(b)$ and $g(x)=g(y)$, then $(g+f)(x)=g(x)+f(x)=g(x)+b=g(y)+f(y)=(g+f)(y)$. Since $g+f$ is a permutation, this implies that $x=y$. 
\end{proof}

Let $f:\grp{G}\to\grp{G}$ and recall that $u(f)=\max\{\#f^{-1}(b)\,:\, b\in\grp{G}\}$. To prove the lower bound in Theorem \ref{th:lb_ub_pres}, take an element $b\in\im{f}$ such that $\#f^{-1}(b)=u(f)$. If $g:\grp{G}\to\grp{G}$ such that $g+f$ is a permutation, then by Lemma \ref{lm:g_1to1_on_preimg}, $g$ is injective on $f^{-1}(b)$, which means $g$ has at least $u(f)$ distinct images. Therefore, $u(f)\le \pres(f)$ since $\pres(f)=\min\{V(g)\,:\, g+f\text{ is a permutation}\}$. 

To prove the upper bound, observe that we can always construct a $g$ that makes $g+f\in\Omega_{\grp{G}}$ as follows. First, let $0\in\im{g}$, and let $g$ map exactly one element from each preimage set $f^{-1}(b)$ for $b\in\im{f}$ to $0$. This makes sure that $\im{f}\subseteq\im{g+f}$. Now both the domain and the codomain have $\#\left(\grp{G}\setminus\im{f}\right)=q-V(f)$ elements left unassigned, so we can pair up each unassigned $x$ of the domain and unassigned $y$ of the codomain and define $g(x)=y-f(x)$. Thus, $\im{g+f}=\grp{G}$, and $V(g)$ is at worst $q-V(f)+1$ if all $y-f(x)$ happens to be distinct. Hence, $\pres(f)\le q-V(f)+1$. This completes the proof of Theorem \ref{th:lb_ub_pres}. 

The lower bound in Theorem \ref{th:lb_ub_pres} may coincide with the upper bound for some $f$, and thus $\pres(f)=u(f)$ when that happens. In the following, we prove that this is precisely when a function is either a permutation, or is mapping everything one-to-one except for exactly one element that has more than one preimage. 
\begin{theorem}\label{th:LB=UB}
Let $f:\grp{G}\to\grp{G}$. Then $u(f)=q-V(f)+1$ if and only if either $f$ is a permutation, or there exists a unique $y\in\im{f}$ for which $\#f^{-1}(y)=u(f)>1$ and $\#f^{-1}(b)=1$ for all $b\in \im{f}\setminus\{y\}$. 
\end{theorem}
\begin{proof}
Clearly, if $f$ is a permutation, then $u(f)=1$ and $q-V(f)=0$ so the assertion is true in this case. We are interested in the case when $f$ is not a permutation, so we may assume that $u(f)>1$ and $q-V(f)>0$. The backward direction is easy to verify. Assume that exactly one element $y\in\im{f}$ has more than one preimage, and everything else from $\im{f}\setminus\{y\}$ has exactly one preimage. Then $u(f)=\#f^{-1}(y)$ by definition, and thus $q=u(f)+(V(f)-1)$ by the fact that the disjoint union of all preimage sets is the domain. 

To prove the forward direction, consider the preimage distribution of $f$, which is defined by the ordered $(u(f)+1)$-tuple $(M_0(f),M_1(f),\dots,M_{u(f)}(f))$, where
\begin{equation*}
M_r(f)=\#\{b\in \grp{G}\,:\, \#f^{-1}(b)=r\},\, \text { for }0\le r\le u(f). 
\end{equation*}
So in particular, $q-V(f)=M_0(f)>0$ by our assumption. By counting the elements of the domain and the codomain respectively, we have 
\begin{equation}\label{eq:q_sum_mr}
q=\sum_{r=1}^{u(f)}rM_r(f),\, \text{ and }q=\sum_{r=0}^{u(f)}M_r(f). 
\end{equation}
Equating both expressions of $q$ in (\ref{eq:q_sum_mr}) and moving all terms of $M_r(f)$ for $r>0$ to one side yields
\begin{equation}\label{eq:sum_m0_1}
\begin{aligned}
M_0(f)&=\sum_{r=1}^{u(f)}(r-1)M_r(f)=\sum_{r=2}^{u(f)}(r-1)M_r(f)\\
&=\sum_{r=2}^{u(f)-1}(r-1)M_r(f)+(u(f)-1)M_{u(f)}(f), 
\end{aligned}
\end{equation}
where the summation $\sum_{r=2}^{u(f)-1}$ in the last expression is considered to be empty if $u(f)=2$. If $u(f)=q-V(f)+1$, we may substitute $u(f)-1=M_0(f)$ into the last term of (\ref{eq:sum_m0_1}) and combine with $M_0(f)$ on the left-hand-side: 
\begin{equation}\label{eq:sum_m0_2}
\begin{aligned}
0&=\sum_{r=2}^{u(f)-1}(r-1)M_r(f)+M_0(f)\left(M_{u(f)}(f)-1\right). 
\end{aligned}
\end{equation}
By definition, all $M_r(f)$'s are nonnegative integers and $M_{u(f)}(f)\ge 1$. So for (\ref{eq:sum_m0_2}) to be true, we must have every term in the summation being $0$. In particular, $M_0(f)(M_{u(f)}(f)-1)=0$, which implies $M_{u(f)}(f)=1$. Moreover, if $u(f)>2$, then $\sum_{r=2}^{u(f)-1}(r-1)M_r(f)=0$ implies that $M_r(f)=0$ for all $2\le r\le u(f)-1$. Therefore, $f$ has the preimage distribution as claimed. 
\end{proof}

The upper bound in Theorem \ref{th:lb_ub_pres} is a worst case estimation, so we expect in most cases $\pres(f)$ to be lower. However, it helps us to relate $\pres(f)$ with $N_s(f)$ as we will show in Section \ref{sc:pres_ns}, where (\ref{eq:q_sum_mr}) will be used again as it is true for any $f:\grp{G}\to\grp{G}$ in general. 

As we will shortly show, the lower bound in Theorem \ref{th:lb_ub_pres}
can be achieved even when $u(f)\neq q-V(f)+1$. Generally, however, $\pres(f)\neq u(f)$.
A simple example comes from the polynomial $f(x)=x^2(1-x)$ over $\ff{5}$.
Here $u(f)=2$ as there are 2 roots and $f(-1)=f(-2)=2$ (while $f(2)=1$).
If $\pres(f)$ were $2$, then there would exist distinct $a, b\in\ff{5}$
satisfying $\ff{5}=\{0+a,0+b,1+a,2+a,2+b\}$. If we now shift each element by
$-a$, we obtain $\ff{5}=\{0,b-a,1,2,2+b-a\}$, implying
$\{b-a,2+b-a\}=\{3,4\}$.
However, $(b-a)-(2+b-a)=-2\notin \{\pm(3-4)\}=\{\pm 1\}$ so this is impossible. 

We now evaluate $\pres(f)$ for two classes of polynomials over finite fields that are not the type of functions described in Theorem \ref{th:LB=UB} (so the lower and upper bounds in Theorem \ref{th:lb_ub_pres} are different). 
Each provides examples where the lower bound of Theorem \ref{th:lb_ub_pres} is met.
We shall prove an upper bound of $\pres(f)$ by finding a suitable $g\in\ffx{q}$ such that $g+f$ is a permutation. We may always assume that $g(0)=0$ since if $g+f$ is a permutation, then so is $g+f-g(0)$, and $V(g(x))=V(g(x)-g(0))$. 

The polynomials we shall consider are the \textit{$p$-polynomials} over $\ff{p^e}$. These are a class of polynomials of the form $\sum_{i=0}^{e-1} a_i x^{p^i}$ for $a_1,\dots,a_{e-1}\in\ff{p^e}$, and can be used to represent all linear
operators on $(\ff{p^e},+)$ when viewed as a vector space over $\ff{p}$. 
\begin{theorem}\label{th:pres_ppoly}
If $L\in\ffx{p^{e}}$ is a $p$-polynomial, then $\pres(L)=u(L)$. 
\end{theorem}
\begin{proof}
Let $L$ be a $p$-polynomial over $\ff{p^{e}}$. Since $L$ is a linear transformation from $\ff{p^{e}}$ to itself, the kernel $L^{-1}(0)$ and the image set $\im{L}$ of $L$ are subspaces of $\ff{p^{e}}$ over $\ff{p}$, say $\#L^{-1}(0)=p^{e-k}$ and $V(L)=p^k$ for some $0\le k\le e$. If $b\in\im{L}$ and $L(x_0)=b$ for some $x_0\in\ff{p^{e}}$, then $L(x)=b$ if and only if $L(x-x_0)=0$. So $L(x)=b$ has the same number of solutions as $L(x)=0$. Hence, $u(L)=p^{e-k}\le \pres(L)$ by Theorem \ref{th:lb_ub_pres}. 

To prove the other side of the inequality, we shall construct a polynomial $g\in\ffx{p^e}$ such that $g+L$ is a PP and $V(g)=p^{e-k}$. Since $\im{L}$ is a subspace of $\ff{p^{e}}$ with $p^k$ elements, it is an additive subgroup of $\ff{p^{e}}$ with $p^{e-k}$ additive cosets in $\ff{p^{e}}$. Let $C_1,\dots,C_{p^{e-k}}$ be the cosets of $\im{L}$, and $c_{i}$ be a representative of $C_i$ for $1\le i\le p^{e-k}$. Since every element of $\im{L}$ has $p^{e-k}$ distinct preimages, we may partition $\ff{p^{e}}$ into $p^{e-k}$ disjoint subsets $P_1,\dots,P_{p^{e-k}}$, each of cardinality $p^k$, containing exactly one element from each of $f^{-1}(b)$ for all $b\in\im{L}$. In this way, $L$ maps each $P_i$ one-to-one onto $\im{L}$. Consider a function $g\in\ffx{p^{e}}$ that maps $P_i$ constantly to $c_i$ for all $1\le i\le p^{e-k}$. Then $(g+L)(P_i)=c_i+\im{L}=C_i$ for $1\le i\le p^{e-k}$. Therefore, $\im{g+L}=\ff{p^{e}}$ and $g+L$ is a PP. Since $V(g)=p^{e-k}$, we conclude that $\pres(L)=p^{e-k}$. 
\end{proof}

Our second class of polynomials we shall determine $\pres$ for are those
connected to the quadratic character of an odd prime field. These also sometimes meet the lower bound of Theorem \ref{th:lb_ub_pres}. 
\begin{theorem}
Let $f(x)=x^{(p-1)/2}\in\ffx{p}$, the quadratic character of an odd prime field. Then
\begin{equation}\label{eq:pres_quad_char}
\pres(f)=
\begin{cases}
u(f)+1=\frac{p+1}{2},&\text{ if }p\equiv 1\pmod 4;\\
u(f)=\frac{p-1}{2},&\text{ if }p\equiv 3\pmod 4.
\end{cases}
\end{equation}
\end{theorem}

\begin{proof}
The function $f(x)=x^{(p-1)/2}\in\ffx{p}$ has $\im{f}=\{0,1,-1\}$, and $\#f^{-1}(1)=\#f^{-1}(-1)=(p-1)/2$. So $u(f)=(p-1)/2\le \pres(f)$ by Theorem \ref{th:lb_ub_pres}. 

First, consider the case that $p\equiv 3\pmod 4$. To prove (\ref{eq:pres_quad_char}), we shall construct a function $g\in\ffx{p}$ such that $g+f$ is a PP and $V(g)=u(f)=(p-1)/2$. Define a subset $S\subseteq\ff{p}$ of cardinality $(p-1)/2$ by 
\[S=\{0\}\cup\{4t-1,4t\,:\, 1\le t\le (p-3)/4\}. \]
Let $g\in\ffx{p}$ be a function such that $\im{g}=S$, $g(0)=0$, and $g$ maps each of $f^{-1}(1)$ and $f^{-1}(-1)$ one-to-one onto $S$. Then $(g+f)(0)=0$ and 
\begin{equation*}
\begin{aligned}
(g+f)(f^{-1}(1))&=\{s+1\,:\, s\in S\}\\
&=\{1\}\cup\{4t,4t+1\,:\, 1\le t\le (p-3)/4\}\\
&=\{1\}\cup\{4,8,\dots,p-3\}\cup\{5,9,\dots,p-2\},
\end{aligned}
\end{equation*}
and
\begin{equation*}
\begin{aligned}
(g+f)(f^{-1}(-1))&=\{s-1\,:\, s\in S\}\\
&=\{-1\}\cup\{4t-2,4t-1\,:\, 1\le t\le (p-3)/4\}\\
&=\{p-1\}\cup\{2,6,\dots,p-5\}\cup\{3,7,\dots,p-4\}. 
\end{aligned}
\end{equation*}
Therefore, $g+f$ is a PP and $\pres(f)=(p-1)/2$. 

Next, consider the case that $p\equiv 1\pmod 4$. We shall first prove that $\pres(f)> (p-1)/2$ by contradiction, and then prove that $\pres(f)\le (p+1)/2$ by exhibiting an example of $g\in\ffx{p}$ such that $g+f$ is a PP and $V(g)=(p+1)/2$. If $\pres(f)$ were $(p-1)/2$, then there exists $g\in\ffx{p}$ such that $g+f$ is a PP, $g(0)=0$, and $\im{g}=\{y_1=0, y_2,\dots,y_{(p-1)/2}\}$ where the $y_k$'s are all distinct and can be viewed as integers in $\{0,1,2,\dots,p-1\}$. This means that $(g+f)(0)=0$, and 
\begin{equation}\label{eq:Fp_quad_1mod4_1}
\ffs{p}=(g+f)(\ffs{p})=\{y_{k}+1\,:\, k=1,\dots,(p-1)/2\}\cup\{y_{k}-1\,:\, k=1,\dots,(p-1)/2\}. 
\end{equation}

Since both sets in the union (\ref{eq:Fp_quad_1mod4_1}) are of cardinality $(p-1)/2$, they must be disjoint and partition $\ffs{p}$. Since $\{y_1\pm 1\}=\{1,p-1\}$, we must have $y_k\pm 1\notin\{0,1,p-1\}$ for $k=2,\dots,(p-1)/2$, which means 
\begin{equation}\label{eq:Fp_quad_1mod4_2}
\{y_k\pm 1\,:\, k=2,\dots,(p-1)/2\}=\{2,\dots,p-2\}, 
\end{equation} 
and 
\[\{y_2,y_3,\dots,y_{(p-1)/2}\}\cap\{0,1,2,p-2,p-1\}=\emptyset.\]
Since $0$ and $p-1$ are not in $\{y_2,y_3,\dots,y_{(p-1)/2}\}$, (\ref{eq:Fp_quad_1mod4_2}) still holds if we view the elements in $\mathbb{Z}$. Summing the elements of both sides of (\ref{eq:Fp_quad_1mod4_2}) in $\mathbb{Z}$ gives \[\sum_{k=2}^{(p-1)/2}(y_k+1)+(y_k-1)=2\sum_{k=2}^{(p-1)/2}y_k=2+3+\dots+(p-2)=\frac{p(p-3)}{2}.\]
This contradicts the fact that $p(p-3)/2$ is odd when $p\equiv 1\mod 4$. Hence, $\pres(f)>(p-1)/2$. 

To prove that $\pres(f)\le (p+1)/2$, let $g\in\ffx{p}$ such that $g(0)=0$ and 
\begin{equation*}
g(f^{-1}(1))=\{0\}\cup\left\{4t-1,4t\,:\, 1\le t\le (p-5)/4\right\}\cup\{p-3\},
\end{equation*}
\begin{equation*}
g(f^{-1}(-1))=\{0\}\cup\left\{4t-1,4t\,:\, 1\le t\le (p-5)/4\right\}\cup\{p-2\}. 
\end{equation*}
So $V(g)=(p+1)/2$. The images of $g+f$ are $(g+f)(0)=0$,  
\begin{equation*}
\begin{aligned}
(g+f)(f^{-1}(1))&=\{1\}\cup\{4t,4t+1\,:\, 1\le t\le (p-5)/4\}\cup\{p-2\}\\
&=\{1\}\cup\{4,8,\dots,p-5\}\cup\{5,9,\dots,p-4\}\cup\{p-2\},
\end{aligned}
\end{equation*}
and
\begin{equation*}
\begin{aligned}
(g+f)(f^{-1}(-1))&=\{-1\}\cup\{4t-2,4t-1\,:\, 1\le t\le (p-5)/4\}\cup\{p-3\}\\
&=\{p-1\}\cup\{2,6,\dots,p-7\}\cup\{3,7,\dots,p-6\}\cup\{p-3\}. 
\end{aligned}
\end{equation*}
Hence, $g+f$ is a PP and $\pres(f)\le(p+1)/2$. Thus, we conclude that $\pres(f)=(p+1)/2$. 
\end{proof}
\section{Invariants of Permutation Resemblance}\label{sc:pres_inv}

We now move to consider potential invariants of $\pres$.
We begin by recalling some of the more widely known relations
defined on functions.

Let $f,g:\grp{G}\to\grp{G}$.
\begin{itemize}
\item $f$ is \textit{right equivalent} to $g$ if
there exists some $\varphi\in\Omega_{\grp{G}}$ such that $f\circ\varphi=g$.
\item Similarly, we say $f$ is \textit{left equivalent} to $g$ if there exists
some $\varphi\in\Omega_{\grp{G}}$ such that $\varphi\circ f=g$.
\end{itemize}
Now let $f,g\in\ffx{q}$.
\begin{itemize}
\item $f$ and $g$ are \textit{extended affine equivalent
(EA-equivalent)} if there exist affine permutations $A_1, A_2$ and an affine
function $A_3$ such that \[g=A_2\circ f\circ A_1+A_3.\]
If $A_3=0$, then $f$ and $g$ are called \textit{affine equivalent}.
\item A more general equivalence, CCZ-equivalence, is defined by Carlet, Chaprin and Zinoviev in \cite{carlet98}
using the graph of $f\in\ffx{q}$: $G_f=\{\big(x,f(x)\big)\,:\, x\in\ff{q}\}$.
The polynomials $f,g$ are called \textit{CCZ-equivalent} if there
exists an affine permutation $A:\ff{q}^{2}\to\ff{q}^{2}$ such that $A(G_f)=G_g$.
CCZ-equivalence is a generalization of EA-equivalence; to see this, we refer
the reader to Budaghyan, Carlet and Pott \cite{budaghyan06}, where several
examples of pairs functions that are CCZ-equivalent but not EA-equivalent are
given.
\end{itemize}
We begin by first considering left and right equivalence.
\begin{lemma}\label{lm:img_righteq}
Suppose $f:\grp{G}\to\grp{G}$ and $\varphi\in\Omega_{\grp{G}}$. Then 
\begin{equation}\label{eq:img_righteq}
\begin{aligned}
\im{f}=\im{f\circ\varphi},
\end{aligned}
\end{equation}
and hence, 
\begin{equation}\label{eq:vf_righteq}
V(f)=V(f\circ\varphi). 
\end{equation}
Moreover, for any $b\in\grp{G}$, 
\[\#f^{-1}(b)=\#(f\circ\varphi)^{-1}(b).\]
\end{lemma}
\begin{proof}
The first assertion follows from the fact that $\varphi(\grp{G})=\grp{G}$.
We have 
\[\im{f}=f(\grp{G})=f(\varphi(\grp{G}))=(f\circ\varphi)(\grp{G})=\im{f\circ\varphi}.\]
The second assertion is because 
\begin{equation*}
\begin{aligned}
(f\circ \varphi)^{-1}(b)&=\{a\in\grp{G}\,:\, f(\varphi(a))=b\}\\
&=\{\varphi^{-1}(c)\in\grp{G}\,:\, f(c)=b\}\\
&=\varphi^{-1}\left(\{c\in\grp{G}\,:\, f(c)=b\}\right).
\end{aligned}
\end{equation*}
Since $\varphi^{-1}$ is also a permutation, we have
\begin{equation*}
\begin{aligned}
\#(f\circ \varphi)^{-1}(b)&=\#\{c\in\grp{G}\,:\, f(c)=b\}=\#f^{-1}(b).
\end{aligned}
\end{equation*}
\end{proof}
\begin{cor}\label{co:right_inv}
Let $f:\grp{G}\to\grp{G}$ and $\varphi\in\Omega_{\grp{G}}$. Then 
\[\pres(f)=\pres(f\circ\varphi),\]
i.e., $\pres$ is a right-invariant. 
\end{cor}
\begin{proof}
Let $f,h:\grp{G}\to\grp{G}$ and $\varphi\in\Omega_{\grp{G}}$. By Lemma \ref{lm:img_righteq}, $V(f-h)=V((f-h)\circ \varphi)=V(f\circ \varphi-h\circ \varphi)$. If $h$ is a permutation over $\grp{G}$, then so is $h\circ \varphi$. Hence, 
\begin{equation*}
\begin{aligned}
\pres(f)&=\min\{V(f-h)\,:\, h\in\Omega_{\grp{G}}\}\\
&=\min\{V(f\circ \varphi-h)\,:\, h\in\Omega_{\grp{G}}\}=\pres(f\circ\varphi). 
\end{aligned}
\end{equation*}
\end{proof}
For left equivalence, we still have 
\begin{equation}\label{eq:vf_lefteq}
V(f)=V(\varphi\circ f) 
\end{equation}
for arbitrary $f:\grp{G}\to\grp{G}$ and $\varphi\in\Omega_{\grp{G}}$, because 
\begin{equation*}
\begin{aligned}
\im{\varphi \circ f}=\{\varphi(f(a))\,:\, a\in\grp{G}\}=\varphi(\im{f}). 
\end{aligned}
\end{equation*}
However, the analogue of (\ref{eq:img_righteq}) is clearly not true in general, unless $f\in\Omega_{\grp{G}}$ as well.
Indeed, if $f$ is not a permutation, then for $c\in\im{f}\neq \grp{G}$, we can always define a permutation $\varphi$ such that $\varphi(c)\in \grp{G}\setminus\im{f}$ and therefore $\im{f}\neq \im{\varphi\circ f}$.
(Indeed, a precise condition for when $\im{f}= \im{\varphi\circ f}$ can be
determined, but we omit it here as it is has no further relevance to
our understanding of permutation resemblance.)
Following from this, it becomes clear that $\pres$ is not a left-invariant.
A small counterexample over $\ff{7}$ is easily constructed.
Set $f(x)=x^2\in\ffx{7}$ and $\varphi=(0)(1)(2345)(6)$ in cycle notation.
It is quickly checked that $(\varphi\circ f)(0)=0,
(\varphi\circ f)(\{1,6\})=1, (\varphi\circ f)(\{3,4\})=3,
(\varphi\circ f)(\{2,5\})=5$.
In particular, $\varphi\circ f$ fixes every element of $\{0,1,3,5\}$, and maps
$a\in \{2,4,6\}$ to $7-a$.
Let $g$ be the function on $\ff{7}$ that maps $\{0,1,3,5\}$ to $0$, and
$\{2,4,6\}$ to $1$.
Then $(\varphi\circ f +g)\{0,1,3,5\}=\{0,1,3,5\}$, and
$(\varphi\circ f +g)(\{2,4,6\})=\{2,4,6\}$, so $\varphi\circ f +g$ is a
permutation.
Hence, $\pres(\varphi\circ f)\le V(g)=2$.
However, $\pres(f)>2$ follows from Theorem \ref{th:planar_2to1_pres_ge2} below.
So $\pres(f)\neq \pres(\varphi\circ f)$. 

If $\varphi\in\Omega_{\grp{G}}$ satisfies $\varphi(x-y)=\varphi(x)-\varphi(y)$
for all $x,y\in\grp{G}$, then using a similar argument as in the proof of
Corollary \ref{co:right_inv}, we can show that $\pres(f)=\pres(\varphi\circ f)$
for $f:\grp{G}\to\grp{G}$. This suggests the following. 
\begin{theorem}\label{th:pres_affine_inv}
For $f\in\ffx{q}$ and affine permutations $A_1,A_2\in\ffx{q}$, we have
\[\pres(f)=\pres(A_1\circ f \circ A_2),\]
i.e., $\pres$ is an affine invariant. 
\end{theorem}
\begin{proof}
Let $f,h\in\ffx{q}$ and $A_1,A_2$ be affine permutations over $\ff{q}$. First, by (\ref{eq:vf_righteq}) and \ref{eq:vf_lefteq}), we have $V(f-h)=V(A_1\circ (f-h) \circ A_2)$. Write $A_1=L_1+c_1,A_2=L_2+c_2$ for linearized permutations $L_1,L_2$ and fixed $c_1,c_2\in\ff{q}$. Then for $x\in\ff{q}$ we have
\begin{equation*}
\begin{aligned}
(A_1\circ (f-h) \circ A_2)(x)&=A_1\big(f(A_2(x))-h(A_2(x))\big)\\
&=L_1\big(f(A_2(x))\big)-L_1\big(h(A_2(x))\big)+c_1\\
&=A_1\big(f(A_2(x))\big)-L_1\big(h(A_2(x))\big).
\end{aligned}
\end{equation*}
Since $\Omega$ is a group under composition, $\Omega=\{h\in\Omega\}=\{L_1\circ h\circ A_2\,:\, h \in \Omega\}$. Therefore, 
\begin{equation*}
\begin{aligned}
\pres(f)&=\min\{V(f-h)\,:\, h\in\Omega\}\\
&=\min\{V(A_1\circ (f-h) \circ A_2)\,:\, h\in\Omega\}\\
&=\min\{V(A_1\circ f \circ A_2-h)\,:\, h\in\Omega\}\\
&=\pres(A_1\circ f \circ A_2). 
\end{aligned}
\end{equation*}
\end{proof}
In general, $\pres(f)$ is not an EA-invariant, nor CCZ-invariant.
For example, let $f(x)=x\in\ffx{p^e}$ and $g(x)=\Trsub{\ff{p^e}}{\ff{p}}(x)$
where $e>1$. Since $g=f+\sum_{i=1}^{e-1}x^{p^{i}}$, the two functions are EA-equivalent, and therefore CCZ-equivalent as well. However, we have $\pres(f)=1$ since it is a permutation, while $\pres(g)>1$ since it is not a permutation. 

\section{Permutation Resemblance and Differential Uniformity}\label{sc:pres_planar}

For a nonzero $a\in \grp{G}$, the \textit{differential operator of $f$ in the
direction of $a$} is defined by 
\begin{equation*}
\Delta_{f,a}:=f(x+a)-f(x), 
\end{equation*}
and the \textit{differential uniformity (DU) of $f$} is defined by 
\begin{equation*}
\delta_{f}:=\max_{a\in \grp{G}^\star}u(\Delta_{f,a}).  
\end{equation*}
The concept of DU was first suggested by Nyberg \cite{nyberg93}.
The lower the DU, the more resistant $f$ is to differential attacks when used as an S-box in a cryptosystem. 
The most desired functions are permutation polynomials over finite fields
with optimal (lowest possible) DU.
When $q$ is even, since $\Delta_{f,a}(x)=b$ if and only if
$\Delta_{f,a}(x+a)=b$, the number of solutions of $\Delta_{f,a}(x)=b$ for any
$(a,b)\in \ffs{q}\times\ff{q}$ is always even.
Consequently, the optimal functions are those which are $2$-DU; these are 
called \textit{almost perfect nonlinear (APN)}.
Examples of APN permutations are known. However, when $q=2^n$ with $n$ even,
only one example of an APN permutation is known -- see Browning, Dillon,
McQuistan and Wolfe \cite{bro10} -- and constructing APN functions and APN
permutations remains a central research problem in the area.
When $q$ is odd, it is possible to construct polynomials that are $1$-DU and
these are called \textit{planar}. Equivalently, a function $f$ is planar if
and only if every differential operator of $f$ is a PP.
That they exist is clear: the polynomial $x^2$ is easily seen to be planar
over any field that is not of characteristic 2.
However, for any $a\in\ffs{q}$, $f(x+a)-f(x)=0$ must have a solution, so it
is impossible for a planar function to be a permutation.
In fact, it is shown in \cite{coulter14c} that if $f$ is planar
over $\ff{q}$, then $V(f)$ is no larger than roughly $q-\sqrt{q}$.
Thus, in odd characteristic, our problem becomes that of constructing
permutations with near-optimal DU.

Our first result of this section shows that when $\grp{G}$ is abelian,
given two functions $f,g$, the DU of the
function $g+f$ is controlled by the DU of $f$ and $V(g)$. 
\begin{theorem}\label{th:diffuni_bound}
If $\grp{G}$ is a finite abelian group and $f,g:\grp{G} \to \grp{G}$, then
\begin{equation*}
\begin{aligned}
\delta_{g+f}\le \delta_{f}\cdot \big(V(g)^{2}-V(g)+1\big). 
\end{aligned}
\end{equation*}
In particular, if $V(g)=\pres(f)$, then 
\begin{equation}\label{eq:diffuni_bound_pres}
\begin{aligned}
\delta_{g+f}\le \delta_{f}\cdot \big(\pres(f)^{2}-\pres(f)+1\big). 
\end{aligned}
\end{equation}
\end{theorem}
\begin{proof}
Let $\grp{G}$ be a finite abelian group and $f,g:\grp{G} \to \grp{G}$. Observe that since $\grp{G}$ is abelian, $\Delta_{g+f,a}(x)=g(x+a)+f(x+a)-g(x)-f(x)=\Delta_{g,a}(x)+\Delta_{f,a}(x)$. So for $b\in \grp{G}$, 
\begin{equation*}
\begin{aligned}
\#\Delta_{g+f,a}^{-1}(b)&=\#\{x\,:\,\Delta_{g+f,a}(x)=b\}\\
&=\#\bigcup_{z\in \grp{G}}\{x\,:\,\Delta_{g,a}(x)=z\text{ and }\Delta_{f,a}(x)=b-z\}\\
&=\sum_{z\in \grp{G}}\#\bigg(\{x\,:\,\Delta_{g,a}(x)=z\}\cap\{x\,:\,\Delta_{f,a}(x)=b-z\}\bigg) \\
&\le \sum_{z\in\im{\Delta_{g,a}}}\#\{x\,:\,\Delta_{f,a}(x)=b-z\}\\
&\le \delta_f\cdot V(\Delta_{g,a})
\end{aligned}
\end{equation*}
The result follows from the fact that $\im{\Delta_{g,a}}=\{g(x+a)-g(x)\,:\, x\in\grp{G}\}$ must be a subset of $\{g_1-g_2\,:\, g_1,g_2\in\im{g}\}$, which has at most $V(g)(V(g)-1)+1$ elements. 
\end{proof}
In the case where $f$ is planar, (\ref{eq:diffuni_bound_pres})
simplifies to $\delta_{g+f}\le \pres(f)^{2}-\pres(f)+1$.
Thus, Theorem \ref{th:diffuni_bound} shows that one can approach the
problem of constructing near optimal DU permutations by looking for planar 
functions with low permutation resemblance. As mentioned, there are generally
multiple $g$ for which $g+f$ is a permutation and $V(g)=\pres(f)$, and
it is easily checked computationally (and not at all surprising) that different
$g$ yield permutations $g+f$ with differing DU.
Based on these observations, we end the section by considering the
permutation resemblance of planar functions.

For a planar function $f$, it is known that $V(f)\ge (q+1)/2$, see any
of \cite{qiu07, kyureghyan08, coulter11} for instance.
Appealing to the upper bound in (\ref{eq:lb_ub_pres}) immediately yields
\begin{lemma}\label{th:pres_ub_planar}
If $f\in\ffx{q}$ is planar, then 
\[\pres(f)\le \frac{q+1}{2}. \]
\end{lemma}
To determine a lower bound, we restrict ourselves to a special type of 
planar function; namely those that are referred to as 2-to-1, which means
that $f(0)=0$ and $f$ is 2-to-1 on non-zero elements.
All known planar function classes consist of such a planar function, and
this restriction has been used in several places to obtain important results,
see for instance \cite{ding06} or \cite{qiu07}.
We first prove a necessary condition when a $d$-to-one function satisfies that $\pres(f)=u(f)$. 
\begin{lemma}\label{lm:when_pres_is_uf}
Let $f\in\ffx{q}$. Assume that $f(0)=0$ and $f$ is $d$-to-one over $\ffs{q}$. If $\pres(f)=u(f)=d$, then there exists a $d$-subset $\{c_1,c_2,\dots,c_d\}\subseteq\ff{q}$, such that for all $x,y\in\ffs{q}$, $f(x)-f(y)\notin\{c_i-c_j\,:\, 1\le i,j\le d, i\neq j\}$. 
\end{lemma}
\begin{proof}
Let $f\in\ffx{q}$, $f(0)=0$ and let $f$ be $d$-to-one over $\ffs{q}$. If $\pres(f)=u(f)=d$, then there is a function $g\in\ffx{q}$ such that $g+f$ is a PP and $V(g)=d$. We may also assume that $g(0)=0$, say $\im{g}=\{c_1=0,c_2,\dots,c_d\}$. For $1\le i\le d$, define $P_i=\{x\in\ffs{q}\,:\, g(x)=c_i\}$, the preimage sets of $g$ over $\ffs{q}$. The sets $P_i$'s form a partition $\ffs{q}$, so $\sum_{i=1}^{d}\# P_i=q-1$. 

By Lemma \ref{lm:g_1to1_on_preimg}, $g$ is one-to-one on each $f^{-1}(b)$ for $b\in\im{f}$. Since for all $b\in\ims{f}$, $\#f^{-1}(b)=d$, $f^{-1}(b)$ must contain exactly one element from each of $P_i$ for $1\le i\le d$. Since $\#(\ims{f})= (q-1)/d$ and $\sum_{i=1}^{d}\# P_i=q-1$, by counting the elements we see that $\#P_i=(q-1)/d$ for all $1\le i\le d$, and every $P_i$ contains exactly one element from each $f^{-1}(b)$ for $b\in\ims{f}$. Consequently, $f(P_i)=\ims{f}$. 

Since $g+f$ is a PP and $(g+f)(0)=0$, we can partition $\ffs{q}$ as follows: 
\begin{equation*}
\begin{aligned}
\ffs{q}&=\bigsqcup_{i=1}^{d} (g+f)(P_i)=\bigsqcup_{i=1}^{d} \{c_i+f(x)\,:\, x\in \ffs{q}\}. 
\end{aligned}
\end{equation*}
In particular, any two sets in this union are disjoint. Thus, for all $x,y\in\ffs{q}$ and $i\neq j$, $c_i+f(x)\neq c_j+f(y)$, which implies
\[f(x)-f(y)\neq c_j-c_i.\]
\end{proof}
\begin{theorem}\label{th:planar_2to1_pres_ge2}
Let $q>5$ be an odd prime or odd prime power. If $f\in\ffx{q}$ is a planar function, $f(0)=0$ and $f$ is two-to-one over $\ffs{q}$, then $\pres(f)>2$. 
\end{theorem}
\begin{proof}
Let $f$ be a two-to-one planar function such that $f(0)=0$. Since $f$ is planar, for $a\in \ffs{q}$, $f(x+a)-f(x)$ is a PP. If $\pres(f)=2$, then by Lemma \ref{lm:when_pres_is_uf}, there are two elements, one of them can be $0$, another one $c\in\ffs{q}$, such that $f(x+a)-f(x)\neq \pm c$ for any fixed $a\in\ffs{q}$ and $x\notin\{0,-a\}$. But since $f(x+a)-f(x)$ is a PP, $f(x+a)-f(x)=\pm c$ must have a solution. Therefore, for all $a\in \ffs{q}$, $0$ has to be a solution of either $f(x+a)-f(x)=c$ or $f(x+a)-f(x)=-c$, which means either $f(a)=c$ or $f(a)=-c$. This implies that $f(\ffs{q})=\{\pm c\}$. However, since $f$ is two-to-one over $\ffs{q}$, $\#f(\ffs{q})=(q-1)/2$. So this is impossible when $(q-1)/2>\#\{\pm c\}=2$. 
\end{proof}

\section{Permutation Resemblance and Some Known Measure of Bijectivity}\label{sc:pres_ns}

In this final section, we show that the upper bound in Theorem \ref{th:lb_ub_pres} can be related to the values $N_s(f)$ as defined in (\ref{eq:def_ns}).
This also shows how we can determine an upper bound estimate for $\pres$ in 
cases where the full knowledge of the values for $N_s(f)$ is not known.

Let $f:\grp{G}\to\grp{G}$. Recall that for $0\le r\le u(f)$, $M_r(f)$ is the number of $b\in\grp{G}$ that has exactly $r$ preimages under $f$, as defined in the proof of Theorem \ref{th:LB=UB}. By counting all possible $s$-tuples in each preimage set, we obtain
\begin{equation}\label{eq:ns_sum}
N_s(f)=\sum_{r=s}^{u}(r)_s M_r(f),
\end{equation}
where $(r)_s=r(r-1)(r-2)\cdots(r-s+1)$. The lower bound of $V(f)$ in terms of $N_s(f)$ given in \cite{coulter14c} implies that 
\begin{equation}\label{eq:vf_lb_n2}
V(f)\ge q-\frac{N_2(f)}{2}.
\end{equation} 
Applying this to the upper bound in Theorem \ref{th:lb_ub_pres} gives
\[\pres(f)\le \frac{N_2(f)}{2}+1.\]
This bound is useful as most likely one only has information of $N_2(f)$. But it can be made stronger if one has further information about $N_s(f)$ for larger $s$. In fact, we can use all $N_s(f)$ for $2\le s\le u(f)$, instead of a single choice of $s$. In the following, we show that $q-V(f)=M_0(f)$ is in fact a sum that involves all $N_s(f)$.  
\begin{theorem}\label{th:m0_sum}
Let $f:\grp{G}\to\grp{G}$. Then 
\begin{equation}\label{eq:pfz_m0}
\begin{aligned}
M_0(f)=\sum_{s=2}^{u(f)}\frac{N_s(f)}{s!}(-1)^{s}. 
\end{aligned}
\end{equation}
In particular, 
\begin{equation}\label{eq:pfz_m0_bounds}
\begin{aligned}
\frac{N_2(f)}{2!}-\frac{N_3(f)}{3!}\le M_0(f)\le \frac{N_2(f)}{2!}, 
\end{aligned}
\end{equation}
where the lower bound holds if and only if $M_r(f)=0$ for all $r\ge 4$, and the upper bound holds if and only if $M_r(f)=0$ for all $r\ge 3$ (both bounds coincide in the later case). 
\end{theorem}
\begin{proof}
Consider the ordinary generating function $P_f(z)\in\cpx[z]$ of the preimage distribution
\begin{equation}\label{eq:pfz_zero}
\begin{aligned}
P_f(z)=\sum_{r=0}^{u(f)}M_r(f) z^{r}.
\end{aligned}
\end{equation}
Then by (\ref{eq:q_sum_mr}), $P_f(1)=q$, and 
\begin{equation*}
\begin{aligned}
\frac{d}{dz}  P_f(z)\bigg|_{z=1}=\sum_{r=1}^{u(f)}rM_r(f) z^{r-1}\bigg|_{z=1}=q.
\end{aligned}
\end{equation*}
Furthermore, by (\ref{eq:ns_sum}), for $2\le s\le u(f)$, 
\begin{equation*}
\begin{aligned}
\frac{d^s}{dz^{s}}  P_f(z)\bigg|_{z=1}=\sum_{r=s}^{u(f)}r(r-1)\cdots(r-s+1)M_r(f) z^{r-s}\bigg|_{z=1}=N_s(f).
\end{aligned}
\end{equation*}
Hence, 
\begin{equation}\label{eq:pfz_one}
\begin{aligned}
P_f(z)=q+q(z-1)+\sum_{s=2}^{u(f)}\frac{N_s(f)}{s!}(z-1)^{s}.
\end{aligned}
\end{equation}
Clearly, $P_f(0)=M_0(f)$ from (\ref{eq:pfz_zero}). Therefore, we obtain (\ref{eq:pfz_m0}) by setting $z=0$ in (\ref{eq:pfz_one}). 

To show the bounds in (\ref{eq:pfz_m0_bounds}), we would like to estimate $M_0(f)$ by truncating a few terms from (\ref{eq:pfz_m0}). Suppose we only take the first $s=2,\dots,c-1$ terms for some $3\le c\le u(f)$. Substituting (\ref{eq:ns_sum}) into the tail of (\ref{eq:pfz_m0}) gives
\begin{equation}\label{eq:pfz_m0_tail}
\begin{aligned}
\sum_{s=c}^{u(f)}\frac{N_s(f)}{s!}(-1)^{s}&=\sum_{s=c}^{u(f)}(-1)^{s}\sum_{r=s}^{u(f)}\binom{r}{s}M_r(f)\\
&=\sum_{r=c}^{u(f)}M_r(f)\sum_{s=c}^{r}\binom{r}{s}(-1)^{s}.
\end{aligned}
\end{equation}
Since $\sum_{s=0}^{r}\binom{r}{s}(-1)^{s}=(1-1)^{r}=0$, 
\begin{equation*}
\begin{aligned}
\sum_{s=c}^{r}\binom{r}{s}(-1)^{s}=-\sum_{s=0}^{c-1}\binom{r}{s}(-1)^{s}.
\end{aligned}
\end{equation*}
This sum is a polynomial of degree $c-1$ in $r$, and the coefficient of $r^{c-1}$ is $\frac{(-1)^{c}}{(c-1)!}$. We claim that the roots of this polynomial are $r=1,2,\dots,c-1$. Indeed, since $\binom{r}{s}=0$ if $r<s$, for any $1\le r_0\le c-1$, 
\begin{equation*}
\begin{aligned}
-\sum_{s=0}^{c-1}\binom{r_0}{s}(-1)^{s}=-\sum_{s=0}^{r_0}\binom{r_0}{s}(-1)^{s}=-(1-1)^{r_0}=0.
\end{aligned}
\end{equation*}
Therefore, 
\begin{equation*}
\begin{aligned}
-\sum_{s=0}^{c-1}\binom{r}{s}(-1)^{s}=\frac{(-1)^{c}}{(c-1)!}\prod_{i=1}^{c-1}(r-i),
\end{aligned}
\end{equation*}
and (\ref{eq:pfz_m0_tail}) equals to
\begin{equation*}
\begin{aligned}
\sum_{s=c}^{u(f)}\frac{N_s(f)}{s!}(-1)^{s}=\frac{(-1)^{c}}{(c-1)!}\sum_{r=c}^{u(f)}M_r(f)\prod_{i=1}^{c-1}(r-i).
\end{aligned}
\end{equation*}
Hence, the coefficient of $M_r(f)$ has the same sign $(-1)^{c}$ for all $c\le r\le u(f)$. Since $M_r(f)\ge 0$ for all $0\le r\le u(f)$, the tail $\sum_{s=c}^{u(f)}\frac{N_s(f)}{s!}(-1)^{s}\ge 0$ if $c$ is even, and $\le 0$ if $c$ is odd. In particular, when $c=3$ and $c=4$, we obtain (\ref{eq:pfz_m0_bounds}). The bounds hold if and only if $\sum_{s=c}^{u(f)}\frac{N_s(f)}{s!}(-1)^{s}=\frac{(-1)^{c}}{(c-1)!}\sum_{r=c}^{u(f)}M_r(f)\prod_{i=1}^{c-1}(r-i)=0$. But since the coefficients of the terms $M_r(f)$ all have the same sign, the sum is $0$ if and only if $M_r(f)=0$ for all $c\le r\le u(f)$. 
\end{proof}

Note that since $V(f)=q-M_0(f)$, (\ref{eq:pfz_m0}) and (\ref{eq:pfz_m0_bounds}) are equivalent to 
\[V(f)=q-\sum_{s=2}^{u(f)}\frac{N_s(f)}{s!}(-1)^{s},\] and \[q-\frac{N_2(f)}{2!}\le V(f)\le q-\frac{N_2(f)}{2!}+\frac{N_3(f)}{3!}.\]
Thus, we obtain another proof of (\ref{eq:vf_lb_n2}) that is different from the one in \cite{coulter14c}. 

By Theorem \ref{th:m0_sum}, the upper bound of Theorem \ref{th:lb_ub_pres} is equivalent to 
\[\pres(f)\le 1+\sum_{s=2}^{u(f)}\frac{N_s(f)}{s!}(-1)^{s}.\]
Furthermore, the proof of Theorem \ref{th:m0_sum} shows that even if we do not have all information of $N_s(f)$, we can still get a weaker bound by truncating the sum at an even $s$. For example, 
\[\pres(f)\le 1+\frac{N_2(f)}{2!}-\frac{N_3(f)}{3!}+\frac{N_4(f)}{4!}.\]

\subsection*{Inequivalence of Permutation Resemblance to previous notions}

We wish to end the paper by dealing with the question of equivalence/connection
of $\pres$ with two previous notions in the literature.

In order to study the overall bijectivity of all differential operators of a function, two papers \cite{carlet07} and \cite{panario11} independently defined two different notions.
These were were later shown by Fu, Feng, Wang and Carlet \cite{fu19} to be equivalent up to addition and
multiplication of constants.
In this appendix we show that, while similar, $\pres$ is not equivalent to these notions.
The original definitions in \cite{carlet07} and \cite{panario11} are for functions between two finite abelian groups $\grp{G}_1$ and $\grp{G}_2$ which are not necessarily the same. Since in most cases we are interested in functions from a group to itself, we will just focus on the case where $\grp{G}_1=\grp{G}_2=\grp{G}$ with $q$ elements. 

First, in \cite{carlet07}, Carlet and Ding introduced the imbalance and derivative imbalance of $f$ when the size of the domain and the codomain are not necessarily equal, which generalized the concept of bijectivity. The \textit{imbalance} of $f$ is defined by the variance of $\#f^{-1}(b)$ when $b$ from the codomain is chosen uniformly random. When $f:\grp{G}\to\grp{G}$, we have
\begin{equation*}
Nb_f=\sum_{b\in\grp{G}}\left(\# f^{-1}(b)-1\right)^2=\sum_{b\in\grp{G}}\left(\# f^{-1}(b)\right)^2-q, 
\end{equation*}
and the \textit{derivative imbalance} of $f$ is defined by 
\begin{equation*}
NB_f=\sum_{a\in\grp{G}^{\star}}Nb_{\Delta_{f,a}}. 
\end{equation*}

Second, in \cite{panario11}, Panario, Sakzad, Stevens and Wang introduced the
ambiguity of $f$, which measures the overall injectivity of $\Delta_{f,a}$. Let $\alpha_i(f)=\# \{(a,b)\in \grp{G}^{\star}\times \grp{G}\,:\, \#\Delta_{f,a}^{-1}(b)=i\}$. The \textit{ambiguity} of $f$ is defined by 
\begin{equation*}
A(f)=\sum_{0\le i\le q}\alpha_i(f)\binom{i}{2}. 
\end{equation*}
By fixing each $a\in \grp{G}^{\star}$ at a time when counting $\alpha_i(f)$, $A(f)$ can also be viewed as the sum $\sum_{a\in\grp{G}^{\star}}$ of \textit{row-$a$-ambiguity}, \[A_{r=a}(f)=\sum_{b\in\grp{G}}\binom{\#\Delta_{f,a}^{-1}(b)}{2}. \]

It was shown in \cite{fu19} that $NB_f$ and $A(f)$ are equivalent. In particular, when $f:\grp{G}\to\grp{G}$, $A(f)=NB_f/2$. In fact, one can also express both of them in terms of $N_2(\Delta_{f,a})$ as follows:
\begin{equation*}
\begin{aligned}
N_2(f)&=\#\{(x,y)\,:\, f(x)=f(y),\, x\neq y\}=\#\{(x,y)\,:\, f(x)=f(y)\}-q\\
&=\sum_{b\in\grp{G}}\#\{(x,y)\,:\, f(x)=f(y)=b\}-q\\
&=\sum_{b\in\grp{G}}\left(\#f^{-1}(b)\right)^2-q\\
&=Nb_f. 
\end{aligned}
\end{equation*}
Hence, $NB_f=\sum_{a\in\grp{G}^{\star}}Nb_{\Delta_{f,a}}=\sum_{a\in\grp{G}^{\star}}N_2(\Delta_{f,a})$. One the other hand, 
\begin{equation*}
\begin{aligned}
A_{r=a}(f)&=\sum_{b\in\grp{G}}\binom{\#\Delta_{f,a}^{-1}(b)}{2}\\
&=\sum_{r=2}^{u(f)}\binom{r}{2}\#\{b\in\grp{G}\,:\, \#\Delta_{f,a}^{-1}(b)=r\}\\
&=\sum_{r=2}^{u(f)}\binom{r}{2}M_r(\Delta_{f,a})\\
&=\frac{N_{2}(\Delta_{f,a})}{2}. 
\end{aligned}
\end{equation*}
Therefore, $A(f)=\sum_{a\in\grp{G}^{\star}}A_{r=a}(f)=\sum_{a\in\grp{G}^{\star}}\frac{N_{2}(\Delta_{f,a})}{2}$. 

From the definitions, an obvious distinction between $\pres(f)$, $NB_f$ and $A(f)$ is that $\pres(f)$ is about the bijectivity of $f$ itself, while $NB_f$ and $A(f)$ are about the bijectivity of the difference operators of $f$. Moreover, even if we are looking at the same function, $\pres(\Delta_{f,a})$ is dependent on all $N_s(\Delta_{f,a})$, while $Nb_{\Delta_{f,a}}$ and $A_{r=a}(f)$ are both based on $N_2(\Delta_{f,a})$. In fact, it is possible to have two functions with different $\pres$ but with the same $N_2$. For example, define two functions $f,h$ over $\ff{7}$ by 
\begin{equation*}
f(x)=
\begin{cases}
0 &\text{ when $x=0,1$},\\
2 &\text{ when $x=2,3$},\\
4 &\text{ when $x=4,5$},\\
6 &\text{ when $x=6$}, 
\end{cases}
\qquad\text{ and }\qquad
h(x)=
\begin{cases}
0 &\text{ when $x=0,1,2$},\\
x &\text{ when $x=3,4,5,6$}.
\end{cases}
\end{equation*}
The preimage distributions are $(M_0(f),M_1(f),M_2(f),M_3(f))=(3,1,3,0)$ and $(M_0(h),M_1(h),M_2(h),M_3(h))=(2,4,0,1)$. Both functions have $N_2(f)=N_2(h)=6$. However, $\pres(h)=3$ by Theorem \ref{th:lb_ub_pres} and \ref{th:LB=UB}, while $\pres(f)=2$ because it can be made into the identity function by adding the function
\begin{equation*}
g(x)=
\begin{cases}
0 &\text{ when $x=0,2,4,6$},\\
1 &\text{ when $x=1,3,5$}. 
\end{cases}
\end{equation*}

\bibliographystyle{amsplain}
\bibliography{alteredbibfile.bib}

\end{document}